\documentclass[11pt]{article}
\usepackage{amsmath}
\usepackage{amsfonts,amssymb,amsbsy}

\usepackage{enumerate}
\usepackage{listings}
\usepackage{algpseudocode}
\usepackage{algorithm}
\usepackage{tikz}
\usepackage{circuitikz}
\usepackage{graphicx}
\usepackage{caption}
\usepackage{subcaption}
\usepackage{graphicx}
\usepackage{color}
\usepackage{textcomp,empheq}


\makeatletter

\@namedef{subjclassname@2010}{%
  \textup{2010} Mathematics Subject Classification}
\makeatother

\newtheorem{thm}{Theorem}[section]

\newtheorem{prop}[thm]{Proposition}

\newenvironment{proof}[1][Proof]{\textbf{#1.} }{\hfill \raisebox{-0.1em}{$\Box$}\\[1.5ex]}

\newtheorem{rem}[thm]{Remark}

\newcommand{\me}{\mathrm{e}}

\newcommand{\D}{\displaystyle}
\begin{document}
%
\title{A simple analysis of flying capacitor converter}
%
%
%

\author{
  S. Kadyrov\footnote{
  S. Kadyrov, Faculty of Engineering and Natural Sciences, Suleyman Demirel University, Kaskelen 040900, Kazakhstan,
  Email: {\it shirali.kadyrov@sdu.edu.kz}
  }, 
  P.S. Skrzypacz\footnote{
  P.S. Skrzypacz, School of Science and Technology,
  Nazarbayev University,
  53 Kabanbay Batyr Ave., Astana 010000 Kazakhstan,
  Email: {\it piotr.skrzypacz@nu.edu.kz}
  }, and
Y.L. Familiant\footnote{
  Y.L. Familiant, Independent Researcher,
 Thiensville, Wisconsin, USA
  }
}


\markboth{IEEE Transactions on Power Electronics}%
{A simple second order low pass circuit analysis}
%



\maketitle

\begin{abstract}
{\bf Purpose - } The paper aims to emphasise how switched systems can be analysed with elementary techniques which require only undergraduate-level linear algebra and differential equations. It is also emphasised how math software can become useful for simplifying analytic complications.

{\bf Design/methodology/approach - }The time domain voltage balance methodology is used for stability analysis. As for deriving formulas for the asymptotic average of both capacitor voltage and inductor current, a new simple analytic method is introduced.

{\bf Findings - } It was shown analytically that the time average of capacitor voltage converges to half of the source voltage. A formula for the time average of the current of the inductor is also computed. As a by-product it was discovered that the period of the current is half of the switching period. Numerical simulations are obtained to illustrate the accuracy of the results.

{\bf Research limitations/implications - } Higher dimensional generalisations could become a bit complicated as stability analysis of higher dimensional exponential matrices are not so easy to handle. On the other hand, the new discovery on the period of the current is more likely to give new insights in handling higher dimensional systems. 

{\bf Practical implications - } Analytical formulas are exact and it helps in accurately modelling flying capacitor converts in practice. 

{\bf Originality/value - } FCC is well studied in engineering society. However, not much is done in obtaining exact formulas using analysis. Also, math software is much used in computation of numerical results and obtaining simulations. In this paper, one more important aspect of math software is emphasised, namely, use symbolic computations in analysis.

\end{abstract}


\section{Introduction}
Many natural phenomena can be described either with continuous-time models or discrete-time models. There are cases when one type of modeling is not accurate and instead an interplay between continuous and discrete time dynamics is essential. 
Such models are called switched systems and the current paper deals with one such system, the flying capacitor converter (FCC). The FCC topology has advantages, such as having a natural balancing property and being able to operate at higher voltages, cf. \cite{St2010}. It can be used both as an AC modulation and DC-DC power converter. The load current high order harmonics imply the fast balancing of the capacitors \cite{MFA97}. In this paper, a simple three-level single-leg flying capacitor converter  is considered and its voltage balance dynamics are studied, see Fig.~\ref{fig_0}. The converter consists of one voltage source $V_{dc}$, four switches, and a capacitor $C$. Moreover, the load is modeled by an inductor $L$ connected in series with a resistor $R$. There are various techniques developed to analyze switch mode power supplies, see e.g. \cite{RRM08, Ki12, KS15, YLP15}. Following \cite{RRM08}, the authors are using the time domain voltage balance methodology for stability analysis. However, for analytically deriving the formulas for the asymptotic average of both capacitor voltage and inductor current, a very simple method is introduced. The main goal in this note is to emphasize how switched systems can be analyzed with elementary techniques which require only undergraduate-level basic linear algebra and ordinary differential equations. It is also emphasized how math software can become useful for simplifying analytic complications. Moreover, the Maple\texttrademark ~software scripts used for algebraic operations and numerical simulations are provided for students' convenience. The methods introduced in this paper show practical applications of basic undergraduate courses. The interesting problems suggested in the conclusion can be presented as an undergraduate capstone project and the results can be published in scientific journals. 
\begin{figure}[hbt!]
	\begin{center}
		\includegraphics[scale=0.4]{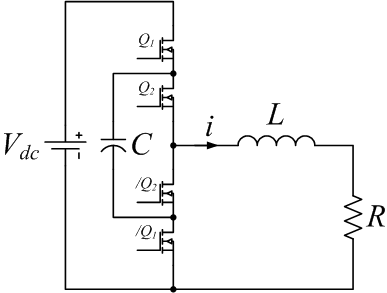}
	\end{center}
	\caption{Flying capacitor DC/DC converter.
		\label{fig_0}}
\end{figure}

\begin{figure}[htb]
\begin{center} 
\begin{subfigure}[htb]{0.35\textwidth}
 \begin{center}
		\includegraphics[scale=0.3]{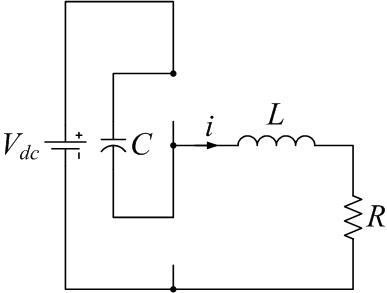}
    \caption{$t\in\left[(k-1)T/2,kT/2\right]$, $k=1,2,3,\ldots$}
  \end{center}
   \end{subfigure}\quad
\hspace*{0.5cm}\begin{subfigure}[htb]{0.35\textwidth}
  \begin{center}
  	\includegraphics[scale=0.3]{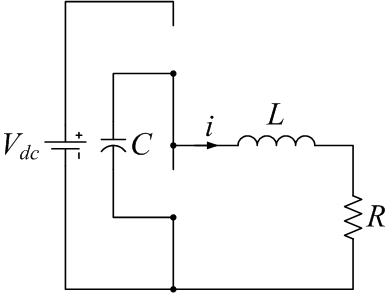}
    \caption{$t\in\left[kT/2,(k+1)T/2\right]$, $k=1,2,3,\ldots$}
  \end{center}
\end{subfigure}
\caption{FC converter switching states}\label{fig_FC_states}
  \end{center}
\end{figure}

\section{Model equation}
Let $\boldsymbol{x}(t):=\left[i(t), v(t)\right]^{^\intercal}$  be the column vector whose unknown components are the inductor current $i(t)$ and the capacitor voltage $v(t)$ at the time $t$, and superscript $^\intercal$ is the transpose. Given positive parameters $R,L,C,V_{dc}$, and $T,$ the authors consider the following non-homogeneous ODE system with periodic piecewise linear coefficients: 
\begin{empheq}[left=\empheqlbrace]{align}
\D\boldsymbol{x}'&=A_1\boldsymbol{x}+\boldsymbol{b}_1 & \text{on}\quad & \left[(k-1)T/2,kT/2\right],\label{eqn_A1_A2_a}\\[1.2ex]
\D\boldsymbol{x}'&=A_2\boldsymbol{x}                  & \text{on}\quad & \left[kT/2,(k+1)T/2\right],\label{eqn_A1_A2_b}\
\end{empheq}
where 
\begin{equation*}
A_1=\left[ \begin{matrix} -\frac{R}{L} & -\frac1L \\[1.5ex] 
\frac1C & 0\end{matrix}\right],~  
A_2=\left[ \begin{matrix} -\frac{R}{L} & \frac1L \\[1.5ex]
 -\frac1C & 0\end{matrix}\right],~ 
\boldsymbol{b}_1=\left[ \begin{matrix} \frac{V_{dc}}{L}  \\[1.5ex]  0\end{matrix}\right]\,,
\end{equation*}
 and $k=1,2,3,\ldots$\,. 
Here, $\D\boldsymbol{x}'=\frac{d\boldsymbol{x}}{dt}$ denotes the time derivative of $\boldsymbol{x}(t)$. The ODE system \eqref{eqn_A1_A2_a}-\eqref{eqn_A1_A2_b} describes the switching states of the FC converter, see Fig.~\ref{fig_FC_states}.\\
The first result is on existence of the periodic solution and its stability.
\begin{thm}\label{thm:main}
The system has a globally asymptotically stable periodic solution. It satisfies
\begin{equation}\label{eqn:xt}
\boldsymbol{x}(t)= \begin{cases}\me^{tA_1}\boldsymbol{x}(0)+ A_1^{-1} (\me^{tA_1} -I)\boldsymbol{b}_1 
& \text{if}~ t\in [0,\frac{T}{2}]\,,\\[1.5ex] 
\me^{(t-\frac{T}{2})A_2}\boldsymbol{x}(\frac{T}{2})& \text{if}~ t\in(\frac{T}{2},T]\,,
\end{cases}
\end{equation}
where 
\begin{equation}\label{eqn:x0}
\boldsymbol{x}(0)= \left(I-\me^{\frac{T}{2} A_2}\me^{\frac{T}{2} A_1}\right)^{-1}
\me^{\frac{T}{2}A_2} A_1^{-1}\left(\me^{\frac{T}{2}A_1} -I\right)\boldsymbol{b}_1.
\end{equation}
\hfill$\diamond$
\end{thm}
Recall that global asymptotic stability means that the continuous solution for any given initial values converges to the periodic solution over time. In particular, the periodic solution is unique.

\begin{proof}
For any initial values $\boldsymbol{x}_0$, it is easy to verify that the unique continuous solution satisfies
\begin{equation}\label{eqn:xtt}
\boldsymbol{x}(t)=\begin{cases}\me^{tA_1}\boldsymbol{x}_0+ A_1^{-1}(\me^{tA_1} -I)\boldsymbol{b}_1 
&~ \text{if}~ t\in [0,\frac{T}{2}]\,,\\[1.5ex] 
\me^{(t-\frac{T}{2})A_2}\boldsymbol{x}(\frac{T}{2})& ~\text{if}~ t\in (\frac{T}{2},T]\,.
\end{cases}
\end{equation}
So, it is sufficient to show that for any initial values $\boldsymbol{x}_0$ there exists a solution which satisfies 
$\boldsymbol{x}(kT) \to \boldsymbol{x}(0)$ as $k \to \infty$, where $\boldsymbol{x}(0)$ is defined as 
in \eqref{eqn:x0}. 
Let 
\[
M:=\me^{\frac{T}{2}A_2}\me^{\frac{T}{2}A_1}
\quad
\text{and} 
\quad
N:=\me^{\frac{T}{2}A_2}A_1^{-1} (\me^{\frac{T}{2}A_1} -I)\,.
\]
Iterating \eqref{eqn:xtt} yields
\begin{align*}
\boldsymbol{x}(kT)&=\me^{\frac{T}{2}A_2}\boldsymbol{x}(kT-T/2)\\
&=\me^{\frac{T}{2}A_2}\bigl\{\me^{\frac{T}{2}A_1}\boldsymbol{x}((k-1)T)
+ A_1^{-1}(\me^{\frac{T}{2}A_1} -I)\boldsymbol{b}_1\bigr\}\\
&=M\boldsymbol{x}((k-1)T)+N\boldsymbol{b}_1\\
&\vdots\\
&=M^k\boldsymbol{x}_0 + \left(\sum_{i=0}^{k-1} M^k\right) N\boldsymbol{b}_1.
\end{align*}
If one can show that the matrix $M$ has all eigenvalues of modulus less than $1$, then 
$\boldsymbol{x}(kT) \to (I-M)^{-1} N\boldsymbol{b}_1=\boldsymbol{x}(0)$ 
as $k \to \infty$. To this end, let $
p(\lambda)=\lambda^2+\alpha \lambda+\beta
$
be the characteristic polynomial of $M$. Then, in view of \cite[Fact~11.18.2]{Be05} it is necessary and sufficient to prove that $|\beta|<1$ and $|\alpha| < 1+\beta$. To this end, set for notational convenience
\[
a=\frac{TR}{2L},\quad b=\frac{T}{2L},\quad c=\frac{T}{2C}\,.
\]
Then, $a,b,c>0$ and by Maple symbolic computing
\[
\beta=e^{-2a} \text{ and } \alpha=\frac{e^{-a} (a^2e^d+a^2e^{-d}+2d^2-2a^2)}{-d^2}, 
\]
where $d:=\sqrt{a^2-4bc}$, (see the Maple script in Algorithm~\ref{alg_2}.)
\begin{algorithm}\caption{Maple script to support the proof of Theorem~\ref{thm:main}.\label{alg_2}
}
{\scriptsize
\begin{verbatim}
restart:
with(LinearAlgebra): 
with(MTM):

A1:=Matrix([[-a, -b], [c, 0]]); 
A2:=Matrix([[-a, b], [-c, 0]]); 
M:=expm(A2).expm(A1); 

P:=CharacteristicPolynomial(M, lambda); 
alpha:=simplify(coeff(P, lambda, 1)); 
beta:=simplify(coeff(P, lambda, 0));
\end{verbatim}
}
\end{algorithm}
Clearly, $|\beta|<1$ and $d^2<a^2$. The following two cases will be considered:\\[1.2ex]
{\it Case 1:}~$d$ is real. Then $a>d>0$ and $\alpha<0$. Then, $|\alpha|=-\alpha<1+\beta$ is equivalent to
\[
a^2e^d+a^2e^{-d}+2d^2-2a^2< d^2e^{a}+d^2 e^{-a}\,,
\] 
which simplifies to $(\frac{1}{d} e^{d/2}-\frac{1}{d}e^{-d/2})^2<(\frac{1}{a}e^{a/2}-\frac{1}{a}e^{-a/2})^2.$ This is obvious since the function $f(t)=\frac{1}{t}\sinh (t)$ is increasing on $t > 0$ and $a>d>0$.\\[1.2ex]
{\it Case 2:}~$d$ is imaginary. Then, $|\alpha|<1+\beta$ if and only of
$$\pm(\frac{1}{d^2}(e^d+e^{-d}-2)+\frac{2}{a^2})< \frac{1}{a^2}(e^a+e^{-a}).$$
As $d$ is imaginary, $e^d+e^{-d}=2\cos(|d|)$, which results in $(e^d+e^{-d}-2)\le 0.$ First, note that
\[
-(\frac{1}{d^2}(e^d+e^{-d}-2)+\frac{2}{a^2})< \frac{1}{a^2}(e^a+e^{-a})
\] 
 is the same as $-\frac{1}{d^2}(e^d+e^{-d}-2)< \frac{1}{a^2}(e^a+e^{-a}+2).$ This is obvious as the left hand side is negative ($d^2<0$), while the right hand side is positive. Finally, it is left to show 
\[
\frac{1}{d^2}(e^d+e^{-d}-2)+\frac{2}{a^2}< \frac{1}{a^2}(e^a+e^{-a})\,,
\]
that is, $\frac{2\cos(|d|)-2}{d^2}<\frac{e^a+e^{-a}-2}{a^2}.$ Notice that 
$
(2\cos(|d|)-2)/d^2=4\sin^2(|d|/2)/|d|^2 \le 1\,.
$
Hence, it is sufficient to prove that $1 < \frac{e^a+e^{-a}-2}{a^2}$, that is, $a/2 < \sinh(a/2)$ 
which is well-known and left to the reader.
%
\end{proof}
Let $\langle f \rangle=\lim\limits_{n \to \infty} \frac{1}{n} \int\limits_0^n f(t)\,dt$ denote the average of $f$. Notice that $\langle v\rangle$ and $\langle i\rangle$  exist because of the global stability due to Theorem~\ref{thm:main}. The following theorem shows how Theorem~\ref{thm:main} 
can be used to find the capacitor voltage and inductor current averages.
\begin{thm}\label{thm:Vave}
Let $[i(t), v(t)]^{^\intercal}$ be a solution of \eqref{eqn_A1_A2_a}-\eqref{eqn_A1_A2_b}. Then, it holds true that
\begin{equation}\label{eq_V_ave}
\langle v \rangle = \frac12 (v(0)+v(T/2))=\frac{V_{dc}}2
\end{equation}
and
\begin{equation}\label{eq_i_ave}
\langle i \rangle = \frac{2C}{T} (V_{dc}-2v(0))=\frac{V_{dc}}{2R}\pm \frac{T}{2RC}\max_{t\in [0,T/2]}|i(t)|\,.
\end{equation}
\hfill $\diamond$
\end{thm}
The first average in Theorem~\ref{thm:Vave} relies on the fact that the period of $i$ is $T/2$. 
Notice that this is rather surprising as all of the coefficients of the ODE system have period $T$ 
and Floquet theory, see e.g. \cite{Ri83}, only guarantees the existence of periodic solutions with period $T$.
The question whether the similar situation appears in the switched systems with more than one capacitor is left for the 
future work.
\begin{prop}\label{prop:halfperiod}
For the initial condition $\boldsymbol{x}(0)$ as in \eqref{eqn:x0}, it holds that
\begin{equation}
i(0)=i(T/2)=i(T)\,.
\end{equation}
\end{prop}
\begin{proof} 
First, notice that $\boldsymbol{x}(0)=\boldsymbol{x}(T)$ implies $i(0)=i(T)$. 
In order to show $i(0)=i(T/2)$, compare $\boldsymbol{x}(0)$ to $\boldsymbol{x}(T/2)$. 
By Theorem~\ref{thm:main},
$\boldsymbol{x}(0)=\boldsymbol{x}(T)=\me^{-\frac{T}{2}A_2}\boldsymbol{x}(T/2)$. 
Thus, it follows from \eqref{eqn:x0} that
\begin{align*}
&\boldsymbol{x}(0)-\boldsymbol{x}(T/2)
= (I - \me^{-\frac{T}{2}A_2})\boldsymbol{x}(0)= (I-\me^{-\frac{T}{2}A_2})(I-\me^{\frac{T}{2}A_2}\me^{\frac{T}{2}A_1})^{-1}\me^{\frac{T}{2}A_2}A_1^{-1}
(\me^{\frac{T}{2}A_1} -I)\boldsymbol{b}_1\,.
\end{align*}
At this point, the authors find it more convenient to work with Maple symbolic solver. Indeed, the first component of the above vector becomes $0$ which means $i(0)=i(T/2)$. The authors find the solution provided by the symbolic solver satisfactory;  
see the Maple script in Algorithm~\ref{alg_1}.
\end{proof}
\begin{algorithm}\caption{Maple script to support the proof of Proposition~\ref{prop:halfperiod}.\label{alg_1}
}
{\scriptsize
\begin{verbatim}
restart: 
with(LinearAlgebra): 
with(MTM):

a:=-R/L; b:=-1/L; c:=1/C;

A1:=Matrix([[a, b], [c, 0]]); 
A2:=Matrix([[a, -b], [-c, 0]]); 
b1:=Matrix([[VDC/L], [0]]); 
Id:=IdentityMatrix(2); 

x:=MatrixInverse(Id-expm((1/2)*T*A2).expm((1/2)*T*A1))
       .expm((1/2)*T*A2).MatrixInverse(A1)
       .(expm((1/2)*T*A1)-Id).b1; 
xT2:=expm((1/2)*T*A1).x + MatrixInverse(A1)
       .(expm((1/2)*T*A1)-Id).b1; 

VAV:=simplify(1/2*(x(2)+xT2(2))); 
IAV:=unapply(2*C*(VDC-2*x(2))/T, T); 

difference:=(Id-expm(-(1/2)*T*A2)).x; 
simplify(difference(1));
\end{verbatim}
}
\end{algorithm}

Now, Theorem~\ref{thm:Vave} can be proved assuming Proposition~\ref{prop:halfperiod}.

\begin{proof}[Proof of Therem~\ref{thm:Vave}]
Since the system is asymptotically globally stable and converges to the unique periodic solution, the average of $\langle v\rangle$ for any initial condition converges to the average of the periodic solution, that is,
\[
\langle v \rangle=\frac{1}{T} \int\limits_0^T v(t)\,dt\,,
\]
where $v(0)$ is given from \eqref{eqn:x0}. Integrating the first row of \eqref{eqn_A1_A2_a} from 0 to $T/2$ and using Proposition~\ref{prop:halfperiod} results in 
\begin{equation}\label{eqn:1stint}
\begin{split}
0&=\int\limits_0^{T/2}i'(t)\,dt= -\frac{R}{L} \int\limits_{0}^{T/2}i(t)\,dt 
- \frac{1}{L}\int\limits_0^{T/2} v(t)\,dt + \frac{V_{dc}T}{2L}\\
&= -\frac{CR}{L} \int\limits_{0}^{T/2}v'(t)\,dt -\frac{1}{L} \int\limits_0^{T/2} v(t)\,dt + \frac{V_{dc}T}{2L}.
\end{split}
\end{equation}
Integrating the first row of \eqref{eqn_A1_A2_b} from $T/2$ to $T$ and using Proposition~\ref{prop:halfperiod} yields
\begin{equation}\label{eqn:2ndint}
\begin{split}
0&=\int\limits_{T/2}^Ti'(t)\,dt = -\frac{R}{L} \int\limits_{T/2}^{T}i(t)\,dt + \frac{1}{L}\int\limits_{T/2}^T v(t)\,dt
= \frac{CR}{L} \int\limits_{T/2}^{T}v'(t)\,dt + \frac{1}{L}\int\limits_{T/2}^T v(t)\,dt\,.
\end{split}
\end{equation}
Consequently,
\begin{equation*}
\begin{split}
T\langle v \rangle &= \int\limits_0^T v(t)\,dt=-CR \int\limits_{0}^{T/2}v'(t)\,dt + \frac{V_{dc}T}2-CR\int\limits_{T/2}^{T}v'(t)\,dt
=\frac{V_{dc}T}2 + CR (v(0)-v(T))=\frac{V_{dc}T}{2}\,.
\end{split}
\end{equation*}
\begin{figure}[htb!]
\begin{center}
\includegraphics[scale=0.45]{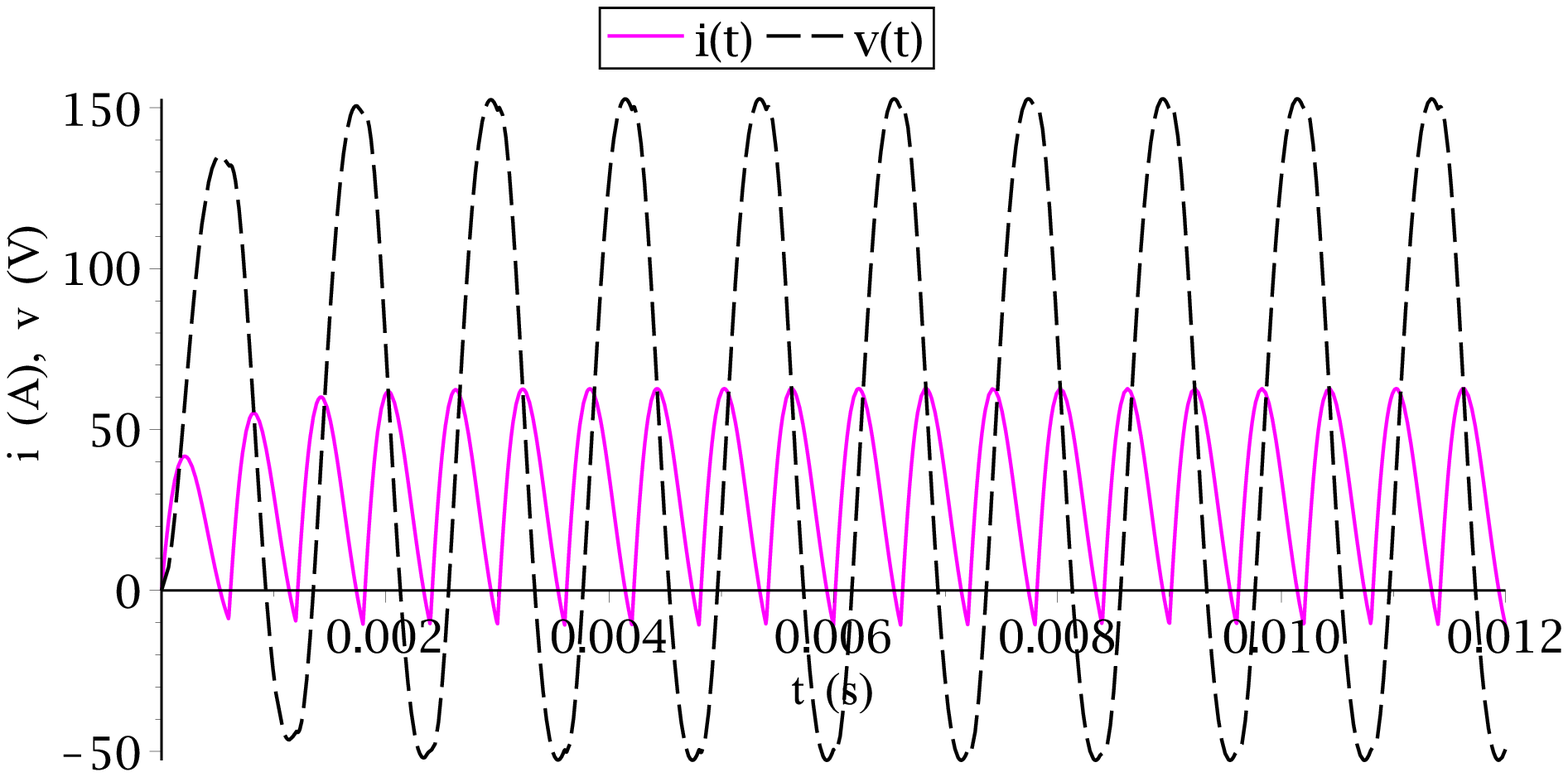}\\
\end{center}
\caption{Current and voltage for $T=1200\cdot 10^{-6} s$, $R=1 \Omega$, $L=0.25\cdot 10^{-3} H$, $C=100\cdot 10^{-6} F$, $V_{dc}=100 V$. 
Averages: $\langle i\rangle = 33.1215 A$, $\langle v\rangle =  50.0000 V$.
\label{fig_1}}
\end{figure}
Multiplying the first and the second rows of \eqref{eqn_A1_A2_a}-\eqref{eqn_A1_A2_b} by $i$ and $v$, respectively, and integrating from $0$ to $T$, results in
\begin{equation}\label{eq:<V>_1}
\int\limits_0^{T/2} v(t)i(t)dt = \int\limits_{T/2}^T v(t)i(t)dt\,,
\end{equation}
and
\begin{equation}\label{eq:<V>_2}
\int\limits_0^{T/2} i^2(t)dt = \frac{V_{dc}}{2R}\int\limits_0^{T/2} i(t)dt\,.
\end{equation}
\begin{algorithm}\caption{Maple script for plotting solutions and calculating the averages.\label{alg_3}
}
{\scriptsize
\begin{verbatim}
restart: 
with(plots):
 
T:=1200*10^(-6);  R:=1;  L:=0.25*10^(-3);  
C:=100*10^(-6);  VDC:=100; 

a:=t->-R/L; 
b:=t->piecewise(t<(1/2)*T, -1/L, 1/L); 
c:=t->piecewise(t<(1/2)*T, 1/C, -1/C); 
b1:=t->piecewise(t<(1/2)*T, VDC/L, 0); 
b_per:=t->b(t-T*floor(t/T)); 
c_per:=t->c(t-T*floor(t/T)); 
b1_per:=t->b1(t-T*floor(t/T)); 

ode1:=diff(i(t),t)=a(t)*i(t)+b_per(t)*v(t)+b1_per(t); 
ode2:=diff(v(t),t)=c_per(t)*i(t); 
ic1:=i(0)=0; 
ic2:=v(0)=0; 
T1:=0; T2:=20*T; 
p:=dsolve({ic1,ic2,ode1,ode2}, numeric, method=rkf45, 
          abserr=10^(-9), maxfun = 500000, range=T1..T2); 

f1:=t->rhs(p(t)[2]); 
f2:=t->rhs(p(t)[3]); 
plot(t->f1(t), t->f2(t)], 0..T2); 
plot(t->f1(t), (T2-2*T)..T2); 

iave=evalf(Int(t->f1(t), (T2-T)..T2))/T; 
iave_theo=VDC/(2*R); 
vave=evalf(Int(t->f2(t), (T2-T)..T2))/T;
vave_theo=VDC/2;
\end{verbatim}
}
\end{algorithm}
Now, integrating only from $0$ to $T/2$ yields   
\begin{equation}\label{eq:<V>_3}
\int\limits_0^{T/2} v(t)i(t)dt=\frac{C}{2}\left(v^2(T/2)-v^2(0)\right)\,
\end{equation}
and by $i(0)=i(T/2)$
\begin{equation}\label{eq:<V>_4}
\begin{split}
&-\frac{R}{L}\int\limits_0^{T/2} i^2(t)dt -\frac{1}{L}\int\limits_0^{T/2}i(t)v(t)dt + \frac{V_{dc}}{L}\int\limits_0^{T/2} i(t)dt\\
&=0\,,
\end{split}
\end{equation}
so 
\begin{equation}\label{eq:<V>_5}
\int\limits_0^{T/2} i(t)dt = \frac{C}{V_{dc}}\left(v^2(T/2)-v^2(0)\right)
\end{equation}
due to \eqref{eq:<V>_2} and \eqref{eq:<V>_3}. Integrating the second row of \eqref{eqn_A1_A2_a} yields
\begin{equation}\label{eq:<V>_6}
\int\limits_0^{T/2} i(t)dt = C\left(v(T/2)-v(0)\right)
\end{equation}
From \eqref{eq:<V>_5} and \eqref{eq:<V>_6} it follows $v(0)+v(T/2)=V_{dc}$. This finishes the proof of \eqref{eq_V_ave}.
Now, the relation \eqref{eq_i_ave} will be proved. From $i(t)=C v'(t)$ on $[0,T/2]$ and $i(t)=-Cv'(t)$ on $[T/2,T]$ it can be deduced that
\[
\langle i \rangle = \frac1T \int\limits_0^Ti(t)\,dt = C (2v(T/2)-v(0) - v(T))\,,
\]
and $v(T)=v(0)=V_{dc}-v(T/2)$ implies the first equality in \eqref{eq_i_ave}. From \eqref{eqn:1stint} and \eqref{eqn:2ndint} it follows that
\begin{equation}\label{eqn:aveTi}
\begin{split}
T \langle i \rangle &= \int\limits_0^Ti(t)\,dt=\frac{V_{dc} T}{2R}+\frac{1}{R} \left(\int\limits_{T/2}^T v(t)\,dt- \int\limits_0^{T/2} v(t)\,dt\right)\,.
\end{split}
\end{equation}
The mean value theorem for integrals implies 
\[
\int\limits_{T/2}^T v(t)\,dt = \frac{T}{2}v(\xi_1)\quad\text{and}\quad 
\int\limits_0^{T/2} v(t)\,dt=\frac{T}2 v(\xi_2)
\]
with some $\xi_1 \in [0,T/2]$ and $\xi_2 \in [T/2,T].$ Then, 
\begin{align*}
v(\xi_1)&=v(T/2)- v'(\eta_1) (T/2-\xi_1)=v(T/2)-\frac{i(\eta_1)}{C}(T/2-\xi_1),\\
v(\xi_2)&=v(T/2)+ v'(\eta_2) (\xi_2-T/2)=v(T/2)- \frac{i(\eta_2)}{C} (\xi_2-T/2),
\end{align*}
for some $\eta_1 \in (\xi_1,T/2)$ and $\eta_2 \in (T/2, \xi_2).$ Thus, \eqref{eqn:aveTi} takes the form
\[
\langle i \rangle =\frac{V_{dc} }{2R}+\frac{1}{2RC} (i(\eta_2)(\xi_2-T/2)-i(\eta_1)(T/2-\xi_1))\,.
\]
\begin{figure}[hbt!]
\begin{center}
\includegraphics[scale=0.45]{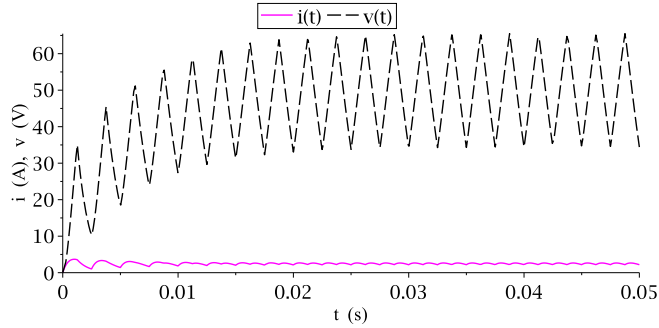}
\end{center}
\caption{Current and voltage for $T=250\cdot 10^{-5} s$, $R=20 \Omega$, $L=10\cdot 10^{-3} H$, $C=100\cdot 10^{-6} F$, $V_{dc}=100 V$. 
Averages: $\langle i\rangle = 2.4922 A$, $\langle v\rangle =  49.9849 V$.
\label{fig_2}}
\end{figure}
Consequently,
\begin{equation*}
\begin{split}
|\langle i \rangle -\frac{V_{dc} }{2R}| &\le|\frac{1}{2RC} (i(\eta_2)(\xi_2-T/2)-i(\eta_1)(T/2-\xi_1))|\le \frac{T}{2RC}\max_{t\in [0,T/2]}|i(t)|\,.
\end{split}
\end{equation*}
\end{proof}
\begin{rem}
The simulation results indicate that $\langle i \rangle \le \frac{V_{dc}}{2R}$; see Table~\ref{table_iave}.
\end{rem}

\section{Numerical experiments}
The computer algebra system Maple was employed to conduct the numerical study; see the Maple script in Algorithm~\ref{alg_3}.
\begin{figure}[hbt!]
\begin{center}
\includegraphics[scale=0.45]{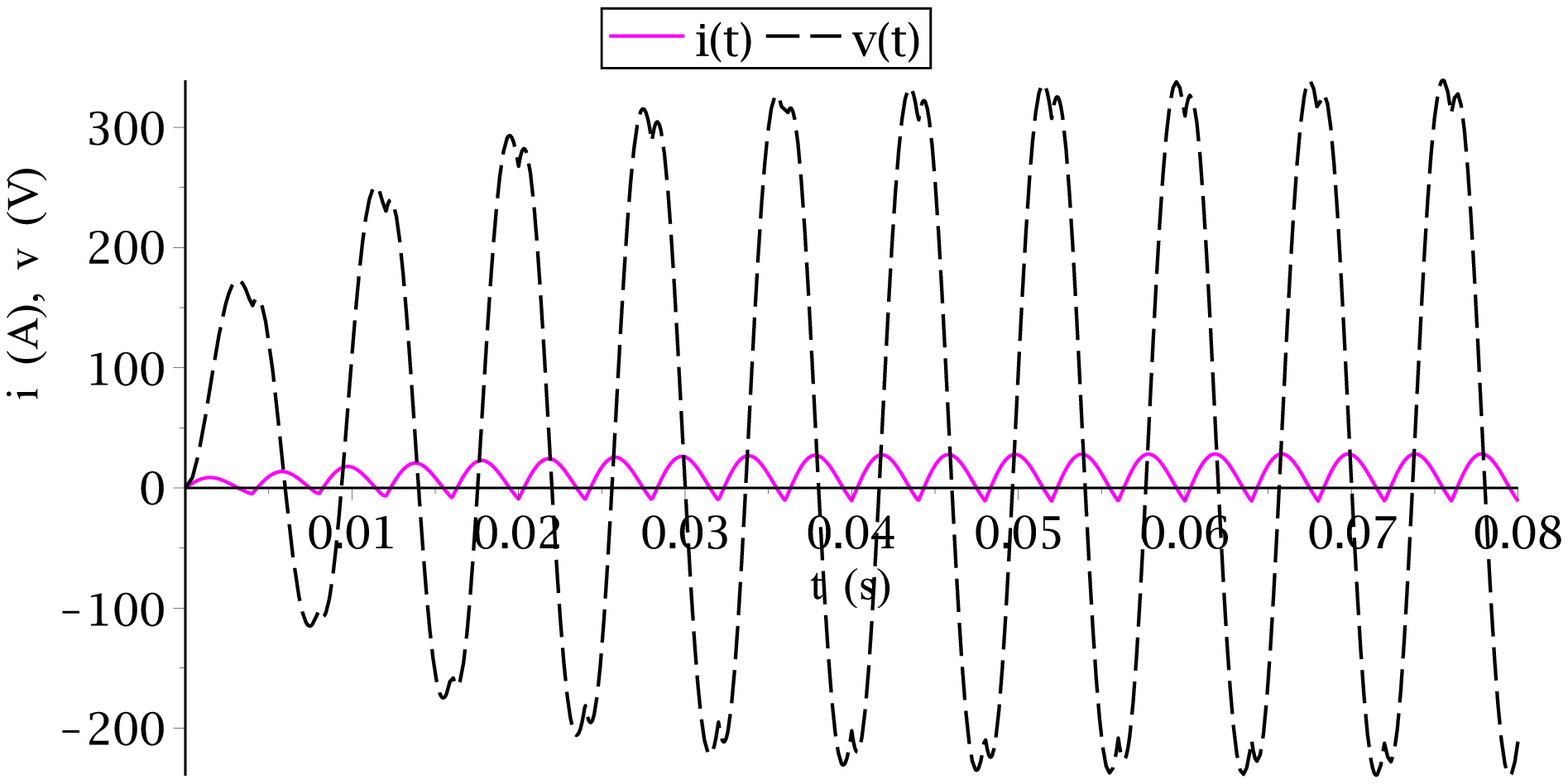}
\end{center}
\caption{Current and voltage for $T=800\cdot 10^{-5} s$, $R=2 \Omega$, $L=10\cdot 10^{-3} H$, $C=100\cdot 10^{-6} F$, $V_{dc}=100 V$. 
Averages: $\langle i\rangle = 13.0181 A$, $\langle v\rangle = 50.0000 V$.
\label{fig_3}}
\end{figure}
The simulation results confirm that the period of the voltage is twice the period of the current. 
The different shapes of profiles for the current and voltage over the normalized two periods are presented in Figure~\ref{fig_4}. The time in Figure~\ref{fig_4} is normalized for each profile due to $t/T_i$ where $T_i$, $i=1,2,3$, is one of the three corresponding periods.
%
\begin{figure}[hbt!]
\begin{center}
~~\includegraphics[scale=0.4]{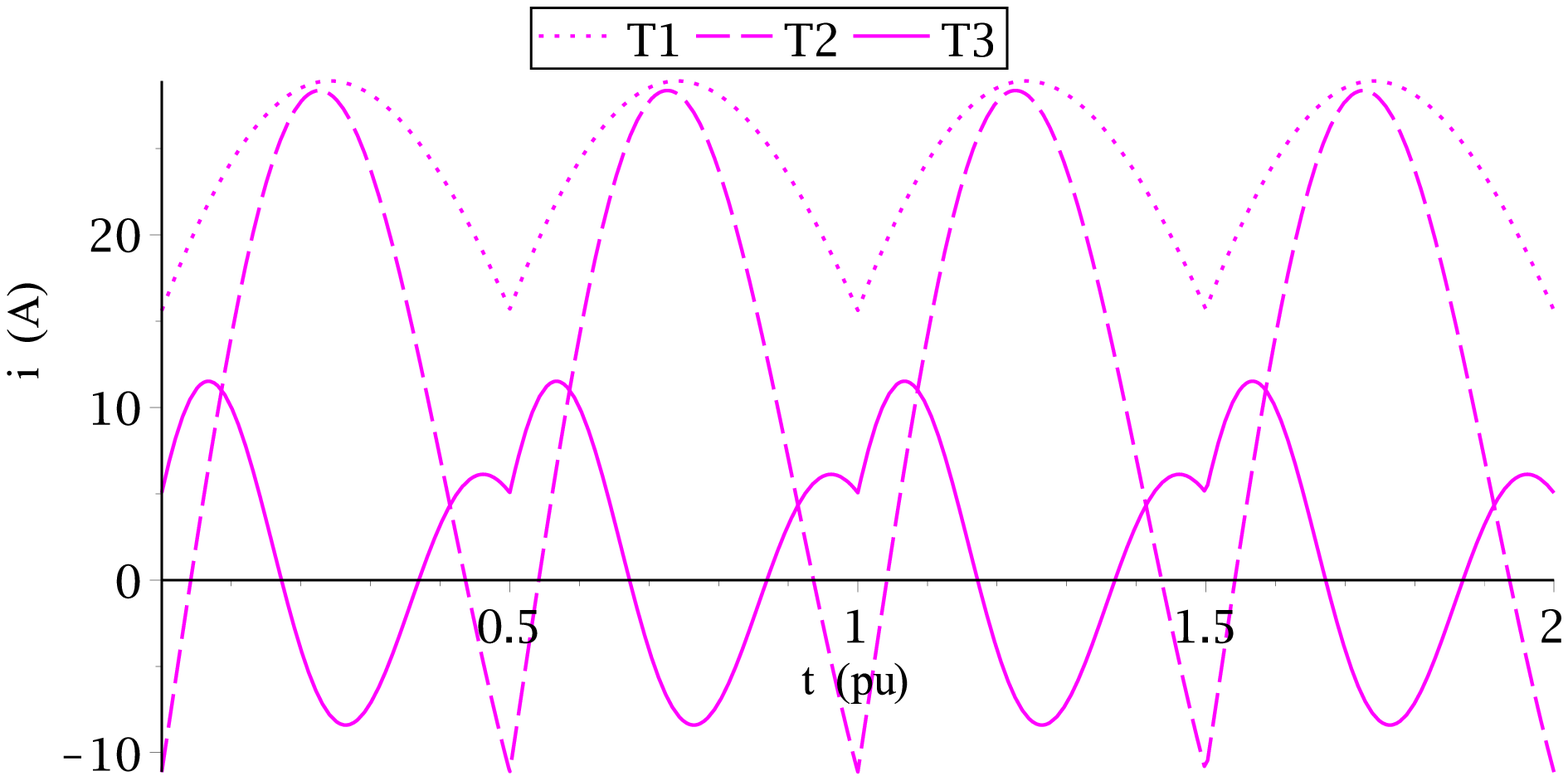}\\[2ex]
\includegraphics[scale=0.4]{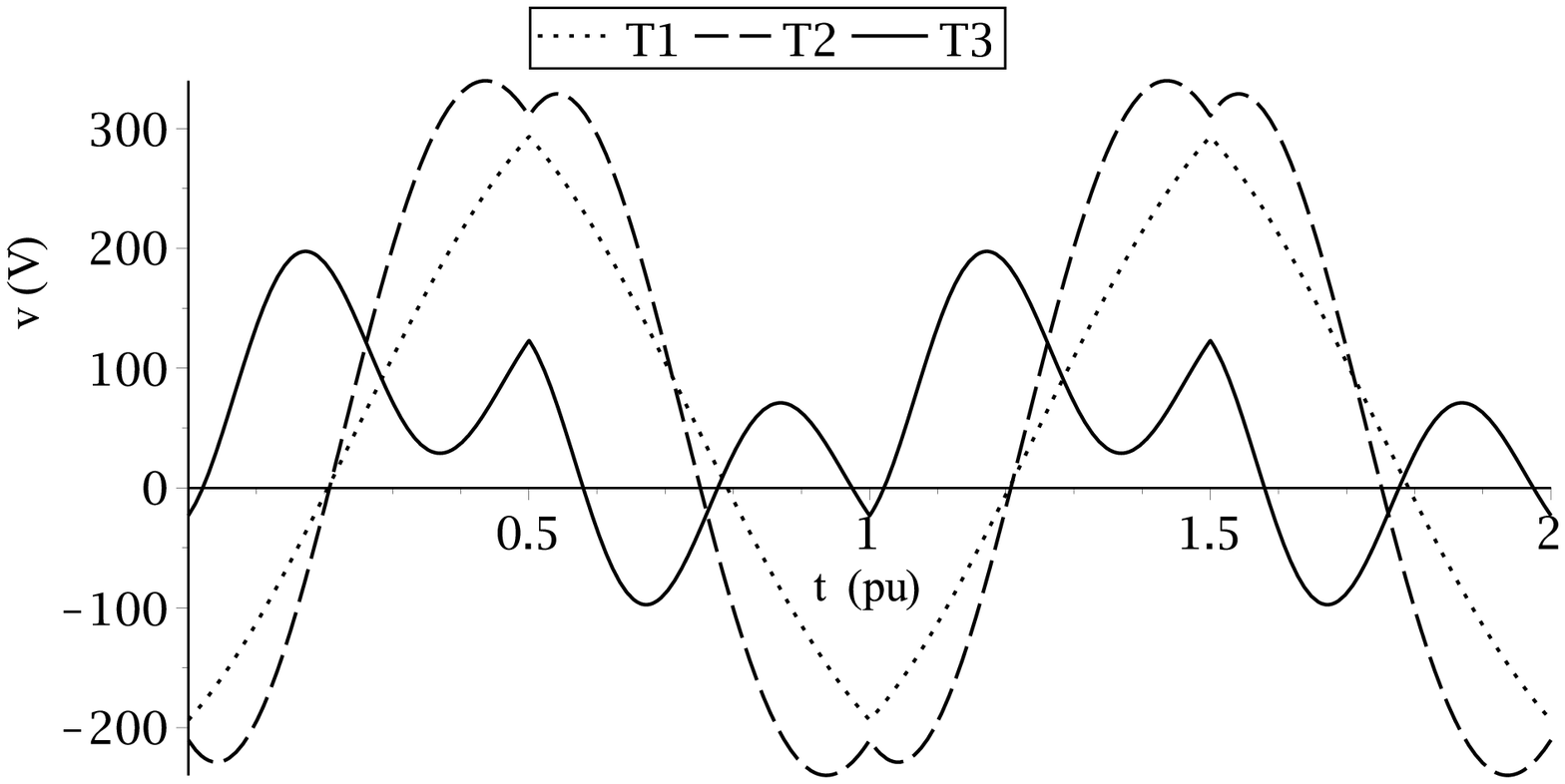}\\[2ex]
\end{center}
\caption{Profiles of current and voltage over two normalized periods for $T_1=400\cdot 10^{-5} s$, $T_2=800\cdot 10^{-5} s$, $T_3=1600\cdot 10^{-5} s$ and
$R=2 \Omega$, $L=10\cdot 10^{-3} H$, $C=100\cdot 10^{-6} F$, $V_{dc}=100 V$.
\label{fig_4}}
\end{figure}
\begin{table}[h!]
\scriptsize
\caption{Averages of current for $\frac{V_{dc}}{2R}=25 A$ and different periods.}\label{table_iave}
\begin{center}
\begin{tabular}{rl|lll}
\hline  
$T$ & in $s$                & $1600\cdot 10^{-5}$ &  $800\cdot 10^{-5}$  & $400\cdot 10^{-5}$\\
$\langle i\rangle$ & in $A$ & $1.8258$ &  $13.0181$  & $24.3412$\\
\hline
\end{tabular}
\end{center}
\end{table}
The numerical average values for the current and voltage are in accordance with theoretical results, see Figures~\ref{fig_1}-\ref{fig_3} and Table~\ref{table_iave}. One can observe that the average of the current deviates from the bound $\frac{V_{dc}}{2R}$ if the period $T$ increases, see Table~\ref{table_iave} and Figure~\ref{fig_5}. Here, equation \eqref{eq_i_ave} from Theorem \ref{thm:main} is used in order to plot the average $\langle i\rangle$ vs. period $T$.
\begin{figure}[htb!]
\begin{center}
\includegraphics[scale=0.4]{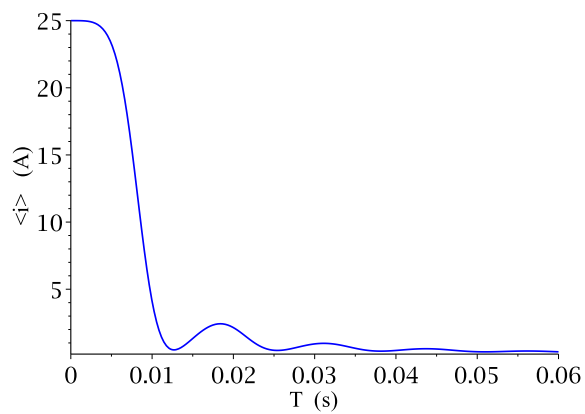}
\end{center}
\caption{Average of current as function of $T$ for $R=2 \Omega$, $L=10\cdot 10^{-3} H$, $C=100\cdot 10^{-6} F$, $V_{dc}=100 V$.
\label{fig_5}}
\end{figure}

\section{Conclusions}
To show the basics of the analysis the authors considered a three-level flying capacitor converter, a simple linear switched system. The periodic switching foreshadows the existence of periodic solutions and this brings two major tasks: a stability analysis of the periodic orbit and a computation of the average value. In this note, the authors used simple methods together with Maple to carry out these tasks. 

There are many interesting problems that can be studied with the introduced techniques. The error term in \eqref{eq_i_ave} hints that the average inductor current can be larger than $\frac{V_{dc}}{2R}$. However, the numerical analysis (see Figure~\ref{fig_5}) suggest that it never exceeds $\frac{V_{dc}}{2R}$. This needs more investigation in the future. In view of equation  \eqref{eq_i_ave}, another interesting task is to find a formula for the amplitude of the current. 

Another general problem is to see how the tools developed in this article can help in analysis of more general multilevel linear switched models. It is especially interesting to find out in which cases the inductor current has frequency twice as large as that of the capacitor voltage.


\section*{Acknowledgment} 
The authors are indebted to Alexander Ruderman from the Department of Electrical Engineering at the Nazabayev University in Astana for providing the problem and fruitful discussions. Helpful comments by Yakov Familiant from Eaton are gratefully acknowledged.


\end{document}